\documentclass[a4paper,UKenglish,cleveref, autoref]{article}
\usepackage{times}  
\usepackage{helvet} 
\usepackage{courier}  
\usepackage{url}  
\usepackage{graphicx}  
\usepackage{hyperref}
\usepackage{comment}
\title{Manipulating Districts to Win Elections: Fine-Grained Complexity}
\author{Eduard Eiben,\textsuperscript{\rm 1} Fedor V. Fomin,\textsuperscript{\rm 2} Fahad Panolan,\textsuperscript{\rm 3} Kirill Simonov\textsuperscript{\rm 2}\\
\textsuperscript{\rm 1} Department of Computer Science, Royal Holloway, University of London, UK\\
\textsuperscript{\rm 2} Department of Informatics, University of Bergen, Norway\\
\textsuperscript{\rm 3} Department of Computer Science and Engineering, IIT Hyderabad, India\\
eduard.eiben@rhul.ac.uk, \{fedor.fomin, kirill.simonov\}@uib.no, fahad.panolan@gmail.com}

\usepackage[utf8]{inputenc}
\usepackage[english]{babel}
\usepackage{amsmath}
\usepackage{amssymb}
\usepackage{amsthm}
\usepackage{thm-restate}
\usepackage{vmargin}
\setmarginsrb{1in}{1in}{1in}{1in}{0pt}{0pt}{0pt}{6mm}

\newtheorem{definition}{Definition}
\newtheorem{claim}{Claim}
\newtheorem{lemma}{Lemma}

\newcommand{\SSS}{\mathcal{S}}
\newcommand{\EEE}{\mathcal{E}}
\newcommand{\Nat}{\mathbb{N}}

\newcommand{\W}{\textsf{W}}
\newcommand{\FPT}{\textsf{FPT}}
\newcommand{\XP}{\textsf{XP}}
\newcommand{\NP}{\textsf{NP}}

\usepackage[obeyFinal,colorinlistoftodos]{todonotes}

\newcommand{\gerrymandering}{\textsc{Gerrymandering}}

\newcommand{\OhOp}[1]{\mathcal{O}\mathopen{}\mathclose\bgroup\left( #1 \aftergroup\egroup\right)}

\DeclareMathOperator{\dist}{dist}
\newcommand{\votes}{v}

\newtheorem{question}{Question}

\begin{document}
\maketitle
 
\begin{abstract}
Gerrymandering   is a practice of manipulating district boundaries and locations in order to achieve  
 a political advantage for a particular party.
Lewenberg, Lev, and Rosenschein [AAMAS 2017]
 initiated the algorithmic study of a geographically-based manipulation problem, where voters must vote at the ballot box closest to them. In this variant of gerrymandering, for a given set of possible locations of  ballot boxes  and  known political preferences of $n$ voters, the task is  to identify locations for  $k$ boxes out of $m$ possible locations to guarantee victory of a certain party in at least $\ell$ districts. Here integers $k$ and $\ell$ are some selected parameter.  
 
It is known that the problem is \NP-complete already for $4$ political parties and prior to our work only heuristic algorithms for this problem were developed. 
We initiate the rigorous study of the gerrymandering problem from the perspectives of parameterized and fine-grained complexity and provide asymptotically matching lower and upper bounds on its computational complexity. We prove that the problem is \W[1]-hard parameterized by $k+n$ and that it  does not admit an $f(n,k)\cdot m^{o(\sqrt{k})}$ algorithm for any function $f$ of $k$ and $n$ only,  unless the Exponential Time Hypothesis (ETH) fails. Our lower bounds hold already for $2$ parties. 
On the other hand, we give an algorithm that solves the problem for a constant number of parties in time $(m+n)^{\OhOp{\sqrt{k}}}$.  
 
\end{abstract}

\section{Introduction}

In 1812, Massachusetts governor Elbridge Gerry immortalized his name by 
signing a bill that created a  mythological salamander shaped district in the Boston area to the benefit of his political party. Since then the term gerrymandering has been used for  the  practice of  establishing a political advantage  by manipulating voting district boundaries. Wikipedia article Gerrymandering contains a lot of notable gerrymandering examples in electoral systems of many countries.  One of possible responses of such manipulations is a more ``rational'' definition of voting districts, for example the system where   voters should always go to a ballot box (or central area) that is closest to them \cite{Terbush:2013yu}.

However,  even in the ``rational'' geographical settings manipulations are still possible. \cite{LewenbergLR17} considered the variant of the gerrymandering problem in which the agent in charge of the design of voting districts has to follow a clear rule: \emph{all voters vote in the ballot box nearest to them.} 
As it was shown by \cite{LewenbergLR17}, even with clear and rational rules, manipulation is possible. 
This brought Lewenberg et al. to the following natural question. 
Is there an efficient algorithmic procedure that would allow an agent, who is obliged to follow the rules, to optimize division for his personal preferred outcome? The good news here is that the existence of such a procedure is highly unlikely, because as it was shown in  \cite{LewenbergLR17}, the problem is \NP-complete already for $4$ political parties.
  
On the other hand, when the number of districts $k$ is small, the problem is trivially  solvable in time $m^k n^{\OhOp{1}}$, where $m$ is the number of all possible locations of ballot boxes and $n$ is the number of voters,  by a brute-force enumerating all possible  locations for $k$ boxes. Thus in the situation when the number of districts is small, there is an efficient procedure for finding an optimal manipulation.  This immediately brings us to the following question which serves as the departure point of our work. 
 
 \begin{question}
	Is it possible to solve the gerrymandering  problem faster than  the brute-force?
\end{question}

To answer this question, we study \gerrymandering{} through the lens of Parameterized Complexity, which is a multivariate paradigm of algorithm analysis introduced by \cite{df99}. In Parameterized Complexity, one measures the performance of algorithms not merely in terms of the size of the input, but also with respect to certain properties of the input or the output (captured by one or several numerical \emph{parameters}, $k$). This gives rise to two notions of tractability, both of which correspond to polynomial-time tractability in the classical setting, when the parameter is constant. First is the class \FPT\ that contains all problems that can be solved in time $f(k)\cdot n^{\OhOp{1}}$ (for some computable function $f$), while the (asymptotically less efficient) class \XP\ contains all problems that can be solved in time $n^{f(k)}$. Aside from these complexity classes, we will make use of the complexity class \W[1], the class of parameterized problems that can be in time $f(k)\cdot n^{\OhOp{1}}$ reduced to \textsc{Independent Set} parameterized by the solution size. \FPT$\neq$\W[1] is the working hypothesis of Parameterized Complexity and it is widely believed that no \W[1]-hard problem is in \FPT. Finally, our more fine-grained lower bound result depends on the Exponential Time Hypothesis (ETH)~\cite{ImpagliazzoPaturi01,LokshtanovMarxSaurabh11},
which states that the satisfiability of {\sc $k$-cnf} formulas (for $k \geq 3$) is not solvable in subexponential-time $2^{o(n)}$, where $n$ is the number of variables in the formula. We refer to the recent books \cite{CyganFKLMPPS15} and \cite{DowneyF13} for more exposition on these complexity notions.

\paragraph{Contribution} \cite{LewenbergLR17} proved that \gerrymandering{$_{plurality}$} (the Gerrymandering problem when the winner at each district is decided by \emph{plurality} voting rule) is \NP-complete, even when the number of candidates $|C|$ is $4$.
We enhance this intractability result by providing a fine-grained complexity of the problem. Also note that our reduction also implies that the problem is NP-complete for any number $|C|\geq 2$ of candidates. 

\begin{restatable}{theorem}{thmlowerbound}\label{thm:lowebound}
For any number $|C|\geq 2$ of candidates, 
\gerrymandering{$_{plurality}$} is \emph{\W[1]-hard} parameterized by $k+n$. Moreover, there is no algorithm solving \gerrymandering{$_{plurality}$} in time $f(k,n)\cdot m^{o(\sqrt{k})}$ for any computable function $f$, unless \emph{ETH} fails. 
\end{restatable}	

The tools we used to obtain lower bounds are based on the grid-tiling technique of \cite{marx-ptaslower}. We refer to the book of \cite{CyganFKLMPPS15} for further discussions of subexponential algorithms. 


Interestingly, the ETH lower bound from Theorem~\ref{thm:lowebound} is asymptotically tight. We give an algorithm whose running time 
matches our lower bound. 

\begin{restatable}{theorem}{thmupperbound}\label{thm:upperbound} For any  constant number of candidates and for any voting rule $g$, 
	$\gerrymandering_g$ is solvable in time $(m+n)^{\OhOp{\sqrt{k}}}$.  
\end{restatable}

To obtain our algorithm, we use the machinery developed by \cite{MarxP15}. The main idea is to guess roughly $\sqrt{k}$ many ballot boxes from an optimal solution that separate the plane into two parts such that each part contains between $k/3$ and $2k/3$ ballot boxes of the solution.

\paragraph{Related work} The \gerrymandering{} problem was introduced by  
\cite{LewenbergLR17}. It is an example of general control problem, where a central agent may influence the outcome using its power over the voting process \cite{BartholdiTT92,borodin2018big,FaliszewskiHH10,HemaspaandraHR07,lev2019reverse}. The problem also belongs to the large family of voting manipulation problems, whose computational aspects are widely studied in the computational social choice community, see \cite{ConitzerW16,FaliszewskiHH15,FaliszewskiP10,XiaZPCR09} for further references. In particular, parameterized complexity of manipulation was studied in \cite{DeyMN15,DeyMN16}.
 
\section{Preliminaries}
For an integer $i$, we let $[i]=\{1,2,\dots,i\}$. 
We denote by $\Nat$ the set of natural numbers and by $\mathbb{R}$ the set of real numbers. Given two points $p_1=(x_1,y_1)$, $p_2=(x_2,y_2)$ in the plane $\mathbb{R}^2$, the \emph{distance} between $p_1$ and $p_2$ is $\dist(p_1,p_2)=\sqrt{(x_2-x_1)^2+(y_2-y_1)^2}$. 
\subsection{Voting and \gerrymandering} 

Given a set of \emph{voters} $V = \{v_1, \ldots ,v_n\}$ and a set of \emph{candidates} $C$, the \emph{preference ranking} of voter $v_i$ is a total order (i.e., transitive, complete, reflexive, and antisymmetric relation) $\succ_i$ over $C$. The interpretation of $c_1\succ_i c_2$ is that the voter $v_i$ strictly prefers candidate $c_1$ over candidate $c_2$. We denote by $\pi(C)$ the set of preference rankings for the set of candidate $C$. A \emph{voting rule} is then a function $f: \pi(C)^n\rightarrow C$. An \emph{election} $\EEE_f= (C,V)$ 
is then comprised of a set of voters $V$, a set of candidates $C$, a preference ranking $\succ_i$ over $C$ for each voter and a voting rule $f$.
In this paper, we will only consider voting rules that do not distinguish between voters. That is the outcome of the voting rule only depends on the number of voters with each preference ranking and not their order. An example of such rule is \emph{plurality}, in which for each voter only the topmost candidate counts and the winner is the candidate that got the most votes.

In a \emph{district-based} election $\EEE_{f,g}$ 
voters are divided into disjoint sets $V_1,\ldots, V_s$ such that $\bigcup_{i\in [s]}V_i=V$. These sets define elections $\EEE_f^i=(C,V_i, )$. The ultimate winner from amongst the winners of $\EEE_f^i$ is determined by voting rule $g$, which for us will be a \emph{threshold function}.

\begin{definition}[\gerrymandering{} \cite{LewenbergLR17}]
	The input of the \gerrymandering{$_f$}  problem is:
	\begin{itemize}
		\item A set of candidates $C$.
		\item A set of $n$ voters $V = \{v_1, \ldots ,v_n\}\subset \mathbb{R}^2$, where every voter $v_i \in V$ is identified by their location on the plane,
		and a preferred ranking $\succ_i\in \pi(C)$.\footnote{Lewenberg et al.  \cite{LewenbergLR17} also define the weighted problem, where every voter $v$ has weight $w_v$. We focus on the version without weights.} 
		\item A set of $m$ possible ballot boxes $B = \{b_1, \ldots, b_m\} \subset \mathbb{R}^2$.
		\item Parameters $\ell, k \in \mathbb{N}$, such that $\ell \le k \le m$.
		\item A target candidate $p \in C$.
	\end{itemize}
	The task is to decide 
	whether
	there is a subset of $k$ ballot boxes $B' \subseteq B$, such that they define a district-based election, in which every voter votes at their closest ballot box in $B'$, the winner at every ballot box is determined by voting rule $f$, and 
	$p$ wins in at least $\ell$ ballot boxes.
\end{definition}
We assume that there are no two possible ballot boxes and a voter such that the voter is equidistant from the two boxes, meaning that no voter could ever be exactly on the boundary between two districts and so a fixed set of ballot boxes unambiguously defines the partition of voters into districts.

In Figure~\ref{fig.ex}, we give an example of an input of \gerrymandering{$_g$}  and two possible solutions. 
\begin{figure}[h]
	\centering
	\includegraphics[width=.6\columnwidth]{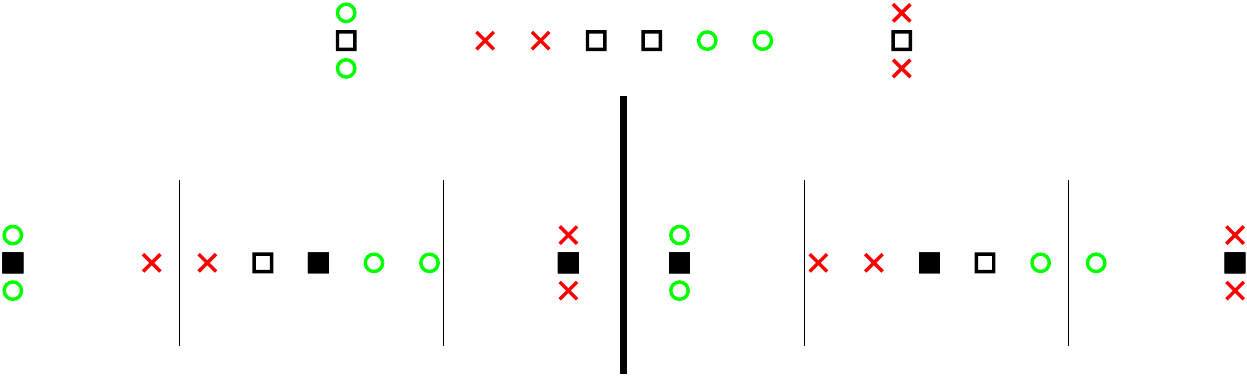}
	\caption{The example of \cite{LewenbergLR17}. The green circle and red cross are voters voting for a different candidates. The squares are ballot boxes. In the bottom are two examples with 3 ballot boxes opened giving a different outcome.}\label{fig.ex}
\end{figure}

\subsection{Voronoi diagrams}

To describe our algorithm, we require the notion of a Voronoi diagram. For a finite set of points $P$ on the plane, \emph{the Voronoi region} of a point $p \in P$ is the set of those points of the plane that are closer to $p$ than to any other point in $P$. We may also use \emph{the Voronoi cell} as a synonym.
A Voronoi region is convex and its boundary is a polyline consisting of segments and possibly infinite rays, since the Voronoi region of $p \in P$ is essentially the intersection of $|P| - 1$ half-planes. \emph{The Voronoi diagram} of $P$ is the tuple of the Voronoi regions of all the elements of $P$.

The Voronoi diagram could be considered a planar graph. The faces of the graph correspond to the Voronoi regions, the edges correspond to the segments which form the boundaries of the individual regions, and the vertices are the endpoints of these segments. The technical details of this correspondence can be found in \cite{MarxP15}.

Voronoi diagrams appear naturally in \gerrymandering\ problem, since when the set $S$ of ballot boxes chosen is fixed, the Voronoi diagram of $S$ defines exactly which regions of the plane are assigned to vote in a particular ballot box.

\section{Hardness}

To obtain our ETH lower bound and prove \W[1]-hardness of \gerrymandering, we reduce from the following \W[1]-hard problem.
\begin{definition}[\textsc{Grid Tilling}~\cite{marx-ptaslower}]\label{def:grid_tiling}
	The input of the \textsc{Grid Tiling} problem is: 
	\begin{itemize}
		\item Integers $k,n\in \Nat$, and a collection $\SSS$ of $k^2$ nonempty sets $S_{i,j}\subseteq [n]\times [n] (1\le i,j\le k)$.
	\end{itemize}
The task is to decide whether there is a set of pairs $s_{i,j}\in S_{i,j}$, for each  $1\le i,j\le k$, such that: 
\begin{itemize}
	\item If $s_{i,j}=(p,q)$ and $s_{i+1,j}=(p',q')$, then $p=p'$.
	\item If $s_{i,j}=(p,q)$ and $s_{i,j+1}=(p',q')$, then $q=q'$.
\end{itemize}
\end{definition}

\begin{lemma}[\cite{marx-ptaslower}, see also Theorem 14.28 in \cite{CyganFKLMPPS15}]
	\textsc{Grid Tiling} is \emph{W[1]-hard} parameterized by $k$ and, unless \emph{ETH} fails, it has no $f(k)n^{o(k)}$-time algorithm for any computable function $f$.
\end{lemma}

\begin{figure}[h]
	\centering
	\includegraphics[width=.6\columnwidth]{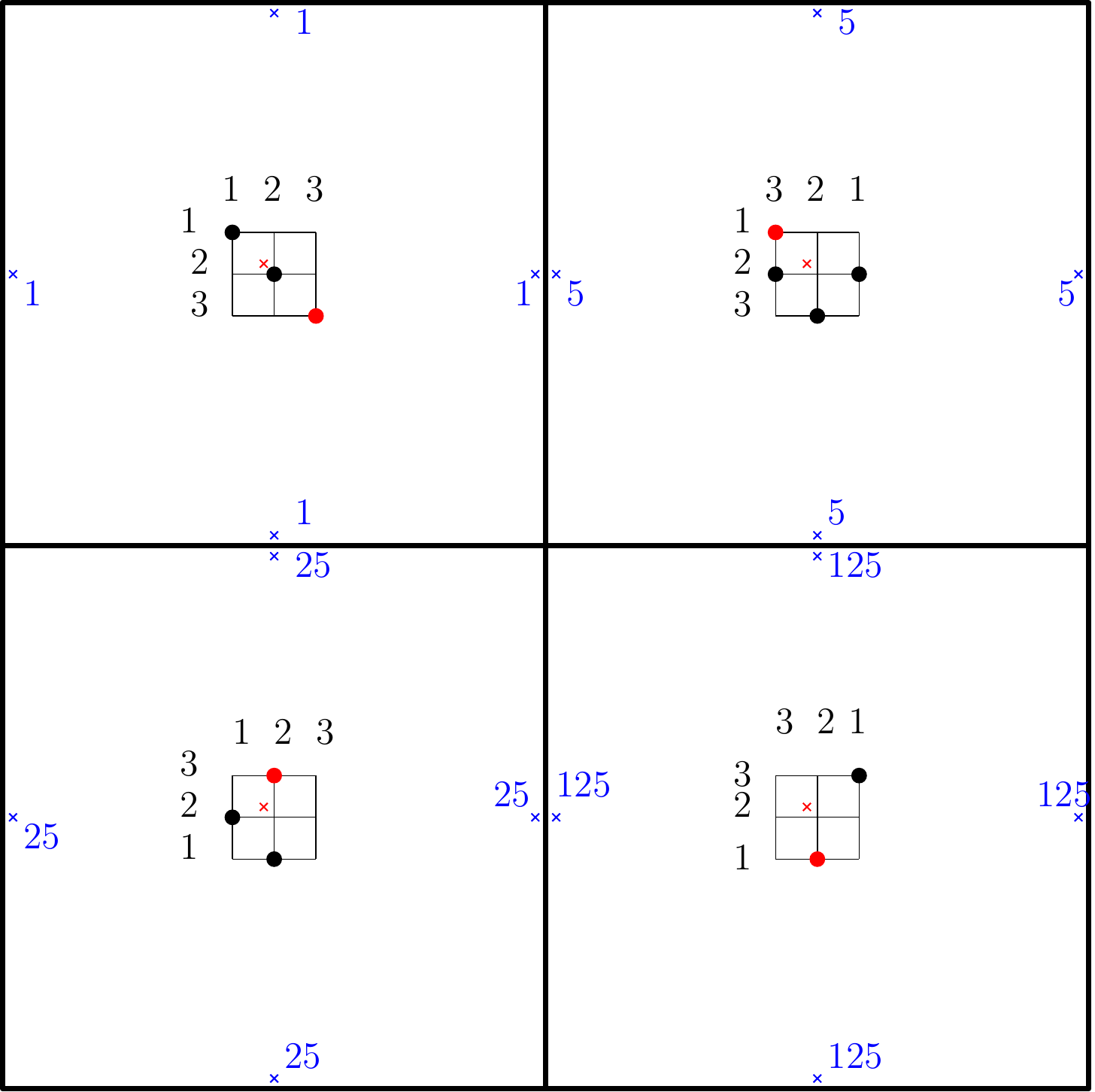}
	\caption{Example of a hardness reduction from Grid Tiling for  $n=3$ and $k=2$. The sets are $S_{1,1}=\{(1,1),(2,2),(3,3)\}$, $S_{1,2}=\{(1,3),(2,1),(2,3), (3,2)\}$, $S_{2,1}=\{(1,2),(2,1),(3,2)\}$, and $S_{2,2}=\{(1,2),(3,1)\}$. The voters are denoted by $\times$, where our candidate is red and the other candidate is blue. The multiplicity of red voters in each cell is $1$ more than the sum of the blue voters in the same cell. The ballot boxes are depicted as $\bullet$. A solution to an instance is the set of red ballot boxes. 
	}\label{fig:hardness}
\end{figure}

\paragraph{Overview of the Proof.} Before we give a formal proof of Theorem~\ref{thm:lowebound}, we first give an informal description of our reduction from \textsc{Grid Tiling} problem that excludes some tedious computations that are necessary for the proof to work, but are not important for understanding where the difficulty of the problem lies. The main idea is to put all voters and ballot boxes inside large $k\times k$ grid such that the cell $(i,j)$ of this grid correspond to the set $S_{i,j}$. In the center of each cell is one much smaller $n\times n$ grid corresponding to the set $S_{i,j}$ (see Figure~\ref{fig:hardness}). The ballot boxes are placed on appropriate positions inside small grid, with a small catch -- for two neighboring cells of large grid, the small grids are ``mirrored''. This way, if we select precisely the same pair $(p,q)$ for each set $S_{i,j}$, then the voters inside each large cell vote in the ballot box inside the same cell. 

We will place a number of the ``other candidate'' voters, determined by a calculation later in the proof, very close to the center of each of 4 borders of the cell (see Figure~\ref{fig:hardness}). By selecting the right distances, for the size of the large grid, size of the small grid, and distance of the ``other candidate'' voters to the border of the cell, we can force that if we select ballot boxes corresponding to $s_{i,j}=(p,q)\in S_{i,j}$ and $s_{i+1,j}=(p',q')\in S_{i+1,j}$ (resp. $s_{i,j+1}=(p',q')\in S_{i,j+1}$), then the voters near the border between cells corresponding to $S_{i,j}$ and $S_{i+1,j}$ (resp. $S_{i,j+1}$) vote inside their corresponding cells if and only if $p=p'$ (resp. $q=q'$). 

Finally, we place in the center of each large cell the number of the ``our candidate'' voters that is just 1 larger than the number of the ``other candidate''  voters inside the cell. As we placed ``our candidate'' voters on $k^2$ different positions, to win $k^2$ ballot boxes, in each box in a solution has to vote all voters from precisely one of these positions. Furthermore, by selecting the number of voters in each cell to grow exponentially, we get that the only way to win all ballot boxes is if all voters in the same cell vote in the same ballot box. This is only possible if we choose in each cell precisely one ballot box and if the corresponding selection of pairs in $S_{i,j}$'s is actually a solution to the original \textsc{Grid Tiling} instance. 

In the remaining of the section we formalize above intuition by assigning the voters and ballot boxes specific positions in the plane and computing distances between positions of ballot boxes and ``other candidate'' voters. 
\thmlowerbound*
\begin{proof}
	
	\newcommand{\redC}{\text{``red''}}
	\newcommand{\blueC}{\text{``blue''}}
	
	We will prove the theorem by a reduction from \textsc{Grid Tiling}. Let $k,n\in \Nat$ be integers and $\SSS$ be a collection of $k^2$ nonempty sets $S_{i,j}\subseteq [n]\times [n] (1\le i,j\le k)$. We will construct an instance $I = (C,V,B,k^2,k^2,\redC)$ of \gerrymandering{$_{plurality}$} with two candidates \redC{} and \blueC{}, 
	$2^{\OhOp{k^2}}$ voters and $\OhOp{k^2n^2}$ ballot boxes and all voters and boxes positioned at the integral points inside $[\OhOp{k\cdot n^2}]\times [\OhOp{k\cdot n^2}]$ grid.


	\noindent\textbf{Ballot boxes.}
	Let $(p,q)\in S_{i,j}$, we let $x_{(p,q)}^{i,j}= (10n^2+4n)(i-1)+5n^2+4p$ if $i$ is even and $x_{(p,q)}^{i,j}= (10n^2+4n)(i-1)+5n^2+4(n-p)$ if $i$ is odd. Similarly, we let $y_{(p,q)}^{i,j}= (10n^2+4n)(j-1)+5n^2+4q$ if $j$ is even and $y_{(p,q)}^{i,j}= (10n^2+4n)(j-1)+5n^2+4(n-q)$ if $j$ is odd. We then identify the pair $(p,q)\in S_{i,j}$ with a ballot box $b_{(p,q)}^{i,j}=(x_{(p,q)}^{i,j}, y_{(p,q)}^{i,j})$. Finally, we will denote by $B_{i,j}$ the set of ballot boxes associated with pairs in $S_{i,j}$.
	
	\noindent\textbf{Voters.} For each pair of $1\le i,j\le k$ we will place $5^{(i-1)\cdot k + (j-1)}$ \blueC\ voters at each of the following four positions: 
	$((10n^2+4n)(i-1)+1, (10n^2+4n)(j-1)+5n^2+2n)$, $((10n^2+4n)(i-1)+5n^2+2n, (10n^2+4n)(j-1)+1)$, 
	$((10n^2+4n)i-1, (10n^2+4n)(j-1)+5n^2+2n)$, and 
	$((10n^2+4n)(i-1)+5n^2+2n, (10n^2+4n)j-1)$. We denote the sets of voters at the above four positions by $V_{\leftarrow}^{i,j}$, $V_\uparrow^{i,j}$, $V_\rightarrow^{i,j}$, and $V_\downarrow^{i,j}$, respectively. And we will place $4\cdot5^{(i-1)\cdot k + (j-1)}+1$ \redC\ voters at the position $((10n^2+4n)+5n^2+2n, (10n^2+4n)(j-1)+5n^2+2n)$, we denote the set of these voters $V_R^{i,j}$.
	
	The goal is to assign $k^2$ ballot boxes such that in each ballot box \redC\ has majority.
	
	\paragraph{Correctness.}
	
	We will show that $(k$, $n$, $\SSS)$ is a YES-instance of \textsc{Grid Tilling} if and only if $(C$, $V$, $B$, $k^2$, $k^2$, $\redC)$ is a YES-instance of \gerrymandering{$_{plurality}$}. Let $s_{i,j}$ be a pair in $S_{i,j}$ for all $i,j\in [k]$ and let $B'=\{b_{s_{i,j}}^{i,j}\mid i,j\in [k] \}$. The following two claims show that if $\bigcup_{i,j\in[k]}\{s_{i,j}\}$ is a solution to \textsc{Grid Tilling}, then $B'$ is a solution to \gerrymandering{$_{plurality}$}. Moreover, if $B'$ is a solution such that for all $i,j\in [k]$ the voters in $\bigcup_{\square\in \{R, \leftarrow,\rightarrow,\uparrow,\downarrow\}}$ vote in $b_{s_i,j}^{i,j}$, then $\bigcup_{i,j\in[k]}\{s_{i,j}\}$ is a solution to \textsc{Grid Tilling}. 
	
	
	\begin{figure}[h]
		\centering
		\includegraphics[width=.8\columnwidth]{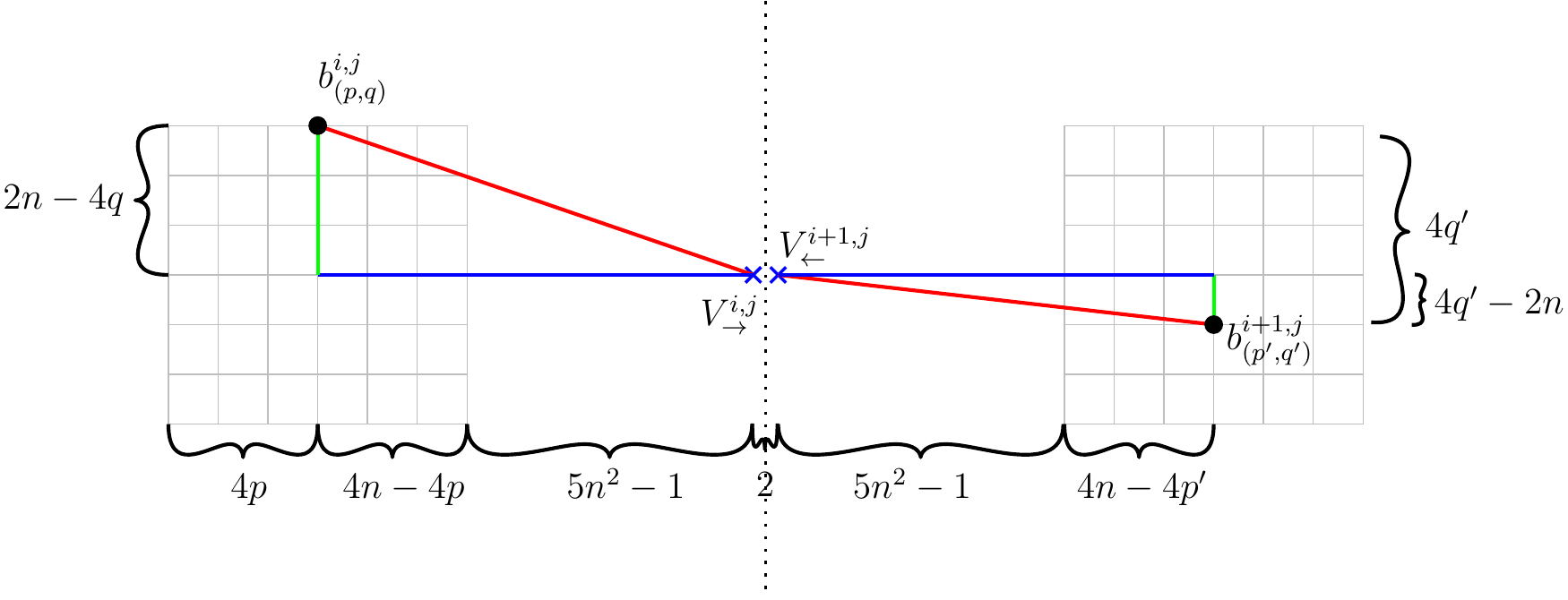}
		\caption{The situation in Claim~\ref{clm:x_axis}. The distance $\dist(b_{(p,q)}^{i,j},V_{\rightarrow}^{i,j})=\sqrt{(5n^2+4n-4p-1)^2+(2n-4q)^2}$. Which is clearly between $5n^2+4n-4p-1$ and $5n^2+4n-4p$. Similarly, $\dist(b_{(p',q')}^{i+1,j},V_{\rightarrow}^{i+1,j})$ is between $5n^2+4n-4p'-1$ and $5n^2+4n-4p'$.
		}\label{fig:hardness2}
	\end{figure}

	\begin{claim}\label{clm:x_axis}
		Let $b_{(p,q)}^{i,j}$ and $b_{(p',q')}^{i+1,j}$ be two ballot boxes. Then for every voter $v\in V_\rightarrow^{i,j}$ it holds $\dist(v,b_{(p,q)}^{i,j})< \dist(v,b_{(p',q')}^{i+1,j})$
		if and only if $p\le p'$. Similarly, for every $v'\in V_\leftarrow^{i+1,j}$ it holds
		$\dist(v',b_{(p',q')}^{i+1,j}) < \dist(v',b_{(p,q)}^{i,j})$
		if and only if $p\ge p'$. 
	\end{claim}
\begin{proof}
	Let us assume that both $i$ and $j$ are even. The remaining $3$ cases are analogous. Then $b_{(p,q)}^{i,j}$ is at position $((10n^2+4n)(i-1)+5n^2+4p, (10n^2+4n)(j-1)+5n^2+4q)$ and $b_{(p',q')}^{i+1,j}$ is at position $((10n^2+4n)i+5n^2+4(n-p'), (10n^2+4n)(j-1)+5n^2+4q')$. Moreover, voters in $V_\rightarrow^{i,j}$ and $V_\leftarrow^{i+1,j}$ are at positions $\left((10n^2+4n)i-1, (10n^2+4n)(j-1)+5n^2+2n\right)$ and $\left((10n^2+4n)i+1, (10n^2+4n)(j-1)+5n^2+2n\right)$, respectively. 
	
	Since both $q$ and $q'$ are numbers between $0$ and $n$, it follows from Pythagorean theorem that the distance between $b_{(p,q)}^{i,j}$ and $V_\rightarrow^{i,j}$ is between $d_{\min}^1=5n^2+4n-4p-1$ and $d_{\max}^1=\sqrt{(d_{\min}^1)^2+(2n)^2}< d_{\min}^1+1$ (see Figure~\ref{fig:hardness2}). For the same reasoning, the distance between $b_{(p,q)}^{i,j}$ and $V_\leftarrow^{i+1,j}$ is between $d_{\min}^1+2$ and $d_{\min}^1+3$, the distance between $b_{(p',q')}^{i+1,j}$ and $V_\leftarrow^{i+1,j}$ is between $d_{\min}^2=5n^2+4n-4p'-1$ and $d_{\min}^2+1$, and finally the distance between $b_{(p',q')}^{i+1,j}$ and $V_\rightarrow^{i,j}$ is between $d_{\min}^2+2$ and $d_{\min}^2+3$. It follows that whenever $p>p'$, then $d_{\min}^2+3< d_{\min}^1$ and the voters in $V_\rightarrow^{i,j}$ are closer to $b_{(p',q')}^{i+1,j}$ than to $b_{(p,q)}^{i,j}$. Similarly, if $p<p'$ then $d_{\min}^1+3< d_{\min}^2$ and the voters in $V_\leftarrow^{i+1,j}$ are closer to $b_{(p,q)}^{i,j}$ than to $b_{(p',q')}^{i+1,j}$.
\end{proof}	

Repeating the same argument for ballot boxes $b_{(p,q)}^{i,j}$ and $b_{(p',q')}^{i,j+1}$, we obtain also the following claim. 

\begin{claim}\label{clm:y_axis}
	Let $b_{(p,q)}^{i,j}$ and $b_{(p',q')}^{i,j+1}$ be two ballot boxes. Then for every voter $v\in V_\downarrow^{i,j}$ it holds $\dist(v,b_{(p,q)}^{i,j})< \dist(v,b_{(p',q')}^{i,j+1})$
	if and only if $q\le q'$. Similarly, for every $v'\in V_\uparrow^{i,j+1}$ it holds
	$\dist(v',b_{(p',q')}^{i,j+1}) < \dist(v',b_{(p,q)}^{i,j})$
	if and only if $q\ge q'$. 
\end{claim}

	From the above claims it directly follows that if we start with a YES-instance of \textsc{Grid Tiling}, we end up with a YES-instance of \gerrymandering. 
	
	For the reverse direction, it remains to show that any solution $B'$ to the reduced instance $I$ of \gerrymandering\ has to select exactly one ballot box $b_{(p,q)}^{i,j}$ in each $B_{i,j}$ and that for each such box, the set of voters voting there is precisely the set of voters associated with $(i,j)$. 
	
	First note that there are precisely $k^2$ positions of \redC\ voters. Hence each ballot box in a solution has to be closest to precisely one of these positions. Moreover, there are only $k^2$ more \redC\ voters than \blueC\ voters. Hence, if \redC\ wins in all boxes, then by pigeon-hole principle \redC\ has to win in each ballot box by precisely $1$. Now let $b\in B'$ be a ballot box closest to $V_R^{i,j}$. Since $V_R^{i,j}$ is precisely the set of \redC\ voters that vote in $b$, none of voters in $\bigcup_{\square\in \{\leftarrow,\rightarrow,\uparrow,\downarrow\}}V_\square^{i',j'}$ such that $(i'-1)k+j'-1>(i-1)k+j-1$ can vote in $b$. Now the number of the remaining \blueC\ voters is precisely $5\cdot 5^{(i-1)k+(j-1)}-1=4\cdot \sum_{\ell=0}^{(i-1)k+(j-1)}5^\ell$.
	Therefore, if one group of voters among $V_{\leftarrow}^{i,j}$, $V_\uparrow^{i,j}$, $V_\rightarrow^{i,j}$, and $V_\downarrow^{i,j}$ does not vote at $b$, then \redC\ wins $b$ by more than $1$ vote and by pigeon-hole principle, there will be a ballot box where \redC\ does not win. Finally, it is easy to see that if there is more than one ballot box selected from $B_{i,j}$, then the box where $V_R^{i,j}$ vote cannot take all four groups of voters $V_{\leftarrow}^{i,j}$, $V_\uparrow^{i,j}$, $V_\rightarrow^{i,j}$, and $V_\downarrow^{i,j}$. Therefore, it follows that in each $B_{i,j}$ exactly one ballot box is selected in $B'$ and the set of voters in this ballot box is precisely the set of voters associated with $(i,j)$. It follows from Claims~\ref{clm:x_axis}~and~\ref{clm:y_axis} that the set of pairs associated with any such solution is a solution to \textsc{Grid Tiling}.  
\end{proof}

\section{Algorithm}

Our algorithm uses the machinery developed by Marx and Pilipczuk \cite{MarxP15} for guessing a balanced separator in the Voronoi diagram of a potential solution.
The main idea of our algorithm is as follows. The Voronoi diagram of $S$ could be viewed as a 3-regular planar graph with $k$ faces and $\OhOp{k}$ vertices.
From \cite{MarxP15} we know that for such a graph there exists a polygon $\Gamma$ which goes through $\OhOp{\sqrt{k}}$ faces and vertices and there are at most $\frac{2}{3}k$ faces strictly inside $\Gamma$ and at most $\frac{2}{3}k$ faces strictly outside $\Gamma$. Moreover, $\Gamma$ could be defined by the sequence of $\OhOp{\sqrt{k}}$ elements of $B$. This allows us to guess all possible variants of $\Gamma$ without actually knowing the solution $S$, and for each $\Gamma$ we recurse into two smaller subproblems defined inside and outside of $\Gamma$. Since the separator $\Gamma$ is balanced and the recurrence $T(k) = n^{\OhOp{\sqrt{k}}} T(\frac{2}{3}k)$ solves to $T(k) = n^{\OhOp{\sqrt{k}}}$, we obtain the desired time bound. See Figure~\ref{fig:voronoi} for the illustration.

\begin{figure}[ht]
    \centering
        \begin{tikzpicture}[auto, scale=0.4, node distance=2cm,every loop/.style={},
            base/.style={draw, fill},
            solution/.style={base, rectangle, black, inner sep=1mm},
            box/.style={base, rectangle, black, inner sep=.5mm},
            inside/.style={blue},
            outside/.style={black},
            sep/.style={red},
            border/.style={draw},
            delta/.style={draw, dashed},
            gamma/.style={draw, thick},
            noose/.style={draw, ultra thick, red, dashed}
        ]
    \node[box] at (3.800, 2.800) (20) {};
    \node[box] at (8.400, 16.300) (21) {};
    \node[box, inside] at (15.500, 12.400) (22) {};
    \node[box, inside] at (10.100, 10.000) (23) {};
    \node[box] at (15.500, 16.000) (24) {};
    \node[box] at (18.700, 3.800) (25) {};
    \node[box, inside] at (3.500, 7.900) (26) {};
    \node[box] at (12.400, 0.900) (27) {};
    \node[box, inside] at (9.900, 11.300) (28) {};
    \node[box] at (9.000, 16.400) (29) {};
    \node[box, inside] at (7.300, 4.900) (30) {};
    \node[box, inside] at (12.900, 12.400) (31) {};
    \node[box] at (2.100, 7.000) (32) {};
    \node[box] at (5.600, 13.100) (33) {};
    \node[box] at (4.100, 0.900) (34) {};
    \node[box] at (2.800, 18.400) (35) {};
    \node[box] at (17.800, 4.200) (36) {};
    \node[box] at (13.800, 15.400) (37) {};
    \node[box] at (18.700, 11.000) (38) {};
    \node[box, inside] at (15.400, 8.500) (39) {};
    \node[solution] at (10.600, 2.300) (0) {};
    \draw[border] (9.594, 5.126) -- (8.710, 0.000) -- (10.857, 0.000) -- (12.286, 3.333) -- (12.838, 5.909) -- (9.594, 5.126);
    \node[solution] at (12.700, 1.400) (1) {};
    \draw[border] (12.286, 3.333) -- (10.857, 0.000) -- (13.714, 0.000) -- (12.286, 3.333);
    \node[solution, sep] at (5.600, 10.700) (2) {};
    \draw[border] (3.110, 10.930) -- (5.700, 9.559) -- (6.789, 10.264) -- (7.815, 12.144) -- (7.131, 13.786) -- (2.872, 11.406) -- (3.110, 10.930);
    \node[solution, sep] at (0.600, 6.800) (3) {};
    \draw[border] (0.000, 8.520) -- (0.000, 0.000) -- (1.446, 0.000) -- (4.234, 4.949) -- (3.418, 6.469) -- (0.000, 8.520);
    \node[solution, inside] at (10.500, 12.200) (4) {};
    \draw[border] (15.184, 13.397) -- (8.033, 11.967) -- (9.850, 10.150) -- (13.242, 9.074) -- (15.184, 13.397);
    \node[solution, inside] at (6.700, 9.000) (5) {};
    \draw[border] (5.700, 9.559) -- (5.700, 5.658) -- (6.983, 5.865) -- (7.456, 7.178) -- (7.291, 9.160) -- (6.789, 10.264) -- (5.700, 9.559);
    \node[solution, inside] at (7.800, 9.500) (6) {};
    \draw[border] (6.789, 10.264) -- (7.291, 9.160) -- (9.383, 9.683) -- (9.850, 10.150) -- (8.033, 11.967) -- (7.815, 12.144) -- (6.789, 10.264);
    \node[solution, sep] at (1.800, 8.800) (7) {};
    \draw[border] (0.000, 11.460) -- (0.000, 8.520) -- (3.418, 6.469) -- (3.110, 10.930) -- (2.872, 11.406) -- (2.363, 11.589) -- (0.000, 11.460);
    \node[solution, sep] at (17.400, 9.100) (8) {};
    \draw[border] (15.184, 13.397) -- (13.242, 9.074) -- (13.520, 6.792) -- (17.750, 6.393) -- (17.750, 18.386) -- (15.184, 13.397);
    \node[solution, sep] at (9.200, 8.100) (9) {};
    \draw[border] (9.383, 9.683) -- (7.456, 7.178) -- (6.983, 5.865) -- (9.594, 5.126) -- (12.838, 5.909) -- (13.137, 6.105) -- (13.520, 6.792) -- (13.242, 9.074) -- (9.850, 10.150) -- (9.383, 9.683);
    \node[solution] at (3.700, 14.100) (10) {};
    \draw[border] (2.657, 14.829) -- (2.363, 11.589) -- (2.872, 11.406) -- (7.131, 13.786) -- (7.958, 17.744) -- (2.657, 14.829);
    \node[solution] at (1.500, 14.300) (11) {};
    \draw[border] (0.000, 16.453) -- (0.000, 11.460) -- (2.363, 11.589) -- (2.657, 14.829) -- (0.000, 16.453);
    \node[solution] at (13.400, 1.700) (12) {};
    \draw[border] (12.286, 3.333) -- (13.714, 0.000) -- (16.800, 0.000) -- (13.137, 6.105) -- (12.838, 5.909) -- (12.286, 3.333);
    \node[solution, sep] at (7.700, 2.800) (13) {};
    \draw[border] (4.234, 4.949) -- (1.446, 0.000) -- (8.710, 0.000) -- (9.594, 5.126) -- (6.983, 5.865) -- (5.700, 5.658) -- (4.234, 4.949);
    \node[solution] at (16.900, 3.800) (14) {};
    \draw[border] (13.137, 6.105) -- (16.800, 0.000) -- (19.900, 0.000) -- (19.900, 5.907) -- (17.750, 6.393) -- (13.520, 6.792) -- (13.137, 6.105);
    \node[solution, inside] at (7.900, 9.100) (15) {};
    \draw[border] (7.291, 9.160) -- (7.456, 7.178) -- (9.383, 9.683) -- (7.291, 9.160);
    \node[solution] at (2.600, 16.100) (16) {};
    \draw[border] (0.000, 16.453) -- (2.657, 14.829) -- (7.958, 17.744) -- (8.897, 19.900) -- (0.000, 19.900) -- (0.000, 16.453);
    \node[solution, inside] at (4.700, 9.000) (17) {};
    \draw[border] (3.110, 10.930) -- (3.418, 6.469) -- (4.234, 4.949) -- (5.700, 5.658) -- (5.700, 9.559) -- (3.110, 10.930);
    \node[solution] at (18.100, 9.100) (18) {};
    \draw[border] (17.750, 18.386) -- (17.750, 6.393) -- (19.900, 5.907) -- (19.900, 19.900) -- (18.458, 19.900) -- (17.750, 18.386);
    \node[solution, sep] at (10.400, 12.700) (19) {};
    \draw[border] (7.131, 13.786) -- (7.815, 12.144) -- (8.033, 11.967) -- (15.184, 13.397) -- (17.750, 18.386) -- (18.458, 19.900) -- (8.897, 19.900) -- (7.958, 17.744) -- (7.131, 13.786);

    \draw[noose] (3) -- (3.416, 6.469) -- (7) -- (3.110, 10.930) -- (2) -- (7.815, 12.144) -- (19) -- (15.184, 13.397) -- (8) -- (13.520, 6.792)-- (9)-- (9.594, 5.126) -- (13) -- (4.234, 4.949) -- (3);

        \end{tikzpicture}
        \caption{The Voronoi diagram of a potential solution. The solution is represented by large squares, small squares are the other possible ballot boxes. A separating polygon $\Gamma$ is in red and dashed, the boxes on $\Gamma$ are also in red. Our algorithm guesses $\Gamma$, takes the red boxes in the solution and then solves the inside (blue and red squares) and the outside (black and red squares) separately.}
        \label{fig:voronoi}
\end{figure}
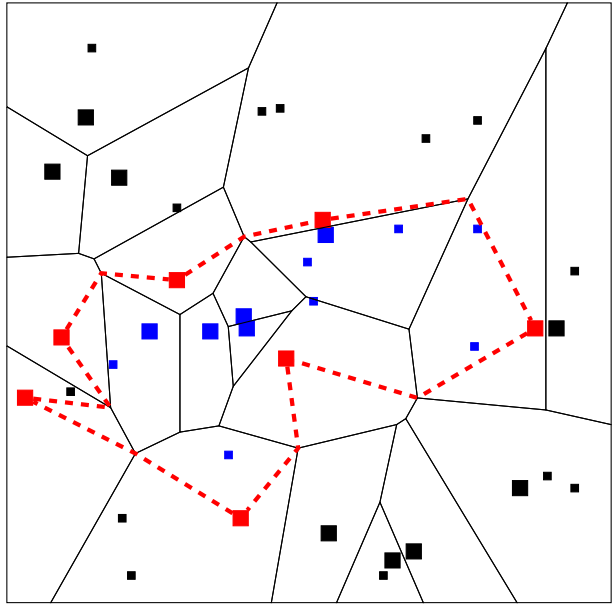

We assume that
no four elements of $B$ lie on the same circle. From this assumption, it follows that no vertex of any Voronoi diagram corresponding to a possible solution has degree more than 3.
We also assume that for any four distinct points in $B$, one is never exactly the center of the unique circle passing through the three others. It means that for any subset $S$ of $B$, the vertices of the Voronoi diagram of $S$ are disjoint from $B$.
Both assumptions are for the ease of presentation and could be achieved by a small perturbation of the coordinates.

Next we state the result of Marx and Pilipczuk about finding small balanced separators in Voronoi diagrams.
\begin{lemma}[\cite{MarxP15}]
    Let $S$ be a set of $k$ points on the plane. Consider the Voronoi diagram of $S$, and the planar graph $G$ associated with it. There is a polygon $\Gamma$ which has length $\OhOp{\sqrt{k}}$ and its vertices alternate between elements of $S$ and vertices of $G$ and each segment of $\Gamma$ lies inside a face of $G$. At most $\frac{2}{3} k$ faces lie strictly inside $\Gamma$, and at most $\frac{2}{3} k$ strictly outside.
    \label{lemma:noose}
\end{lemma}


Note that since $\Gamma$ passes through $\OhOp{\sqrt{k}}$ faces, the number of faces which lie non-strictly inside (outside) $\Gamma$ is bounded by $\frac{3}{4} k$ when $k$ is greater than $\gamma$, where $\gamma$ is some constant.

Now we are ready to prove Theorem~\ref{thm:upperbound} which we restate for convenience.
\thmupperbound*
\begin{proof}
Our algorithm is a recursive brute-force search. An input to the recursion step is a particular \emph{state} $R = (V$, $B$, $\ell$, $k$, $F$, $v$, $\Delta)$. A state consists of a set of voters $V$, a set of possible ballot boxes $B$, parameters $\ell$, $k$ such that $\ell \le k \le |B|$, a subset of boxes $F \subset B$, for any $f \in F$ and $\sigma \in \pi(C)$ there is an integer $\votes(f, \sigma)$ from $0$ to $|V|$, and a set of segments $\Delta$, the segments form the boundary of the region containing $V$ and $B$. Each segment $\delta \in \Delta$ has exactly one endpoint in $F$, and we denote this endpoint by $\delta_F$.
The algorithm returns whether the given state is \emph{valid}.
\begin{definition}[Valid state]
    A state is valid if there exists a subset $S \subset B \setminus F$ such that $|S| + |F| = k$, and in the election with voters $V$ and boxes $S \cup F$ the target party $p$ wins in at least $\ell$ boxes of $S$, and for each $f \in F$ and $\sigma \in \pi(C)$ there are exactly $\votes(f, \sigma)$ voters with preference list $\sigma$ voting in box $f$. Additionally, in the Voronoi diagram of $S \cup F$, each segment $\delta$ of $\Delta$ lies completely inside the cell (touching the border) corresponding to the box $\delta_F$.
    \label{def:valid}
\end{definition}
If the state is valid, the algorithm also returns a particular subset $S \subset B \setminus F$ satisfying Definition~\ref{def:valid}.

Intuitively, $V$ and $B$ correspond to a set of voters and possible boxes in the current subtask, $F$ corresponds to a set of boxes which we have already decided to take into the solution on the previous steps of the recursion, and for any box in $F$ the number of voters with a particular preference list is also fixed through $\votes$. Among the boxes of $B \setminus F$ we are supposed to select at most $k - |F|$ to go into the solution, and in at least $\ell$ of them the target party should win. The condition that segments of $\Delta$ lie inside the corresponding cells enforces that these cells actually separate the whole state from the outside of the region defined by $\Delta$.

Initially, to run our algorithm having the input $(C, V, B, \ell, k, p)$ to $\gerrymandering_f$, we proceed as follows. First, find an equilateral triangle $T$ such that $V$ and $B$ are completely inside $T$. Then mirror each vertex of the triangle against the opposing side, and denote the resulting set of three points as $F$.
Consider the Voronoi diagram of $B \cup F$, by construction the outer faces are exactly the cells of $F$, and these faces are disjoint from $B$ and $V$. Construct a polygon $\Delta$ in the following way. Start from a point of $F$ and go to a next clockwise point of $F$ by a sequence of two segments having as the common point a vertex of the diagram on the border between the corresponding outer faces; repeat until the polygon reaches the starting point again.
The initial state is defined as $R_0 = (V$, $B \cup F$, $\ell$, $k + 3$, $F$, $v$, $\Delta)$ where $v(f, \sigma) = 0$ for any $f \in F$ and $\sigma \in \pi(C)$. Since for any $v \in V$ any $b \in B$ is closer to $v$ than any $f \in F$, $R_0$ is valid if and only if $(C, V, B, \ell, k, p)$ is a YES-instance of $\gerrymandering_f$.

\noindent\textbf{The recursive step.} On the given state $R = (V$, $B$, $\ell$, $k$, $F$, $v$, $\Delta)$ the algorithm proceeds as follows. If $k - |F|$ is at most $\gamma$, where $\gamma$ is the constant mentioned below the statement of Lemma~\ref{lemma:noose}, we try all possible $S \subset B$ of size $k - |F|$ and for each check the conditions of validity in polynomial time, since $\gamma = O(1)$ the whole procedure works in polynomial time as well. If $k - |F|> \gamma$ , try all possible polygons $\Gamma$ of the form as in Lemma~\ref{lemma:noose}, that is, alternating between elements of $B$ and potential vertices of the Voronoi diagram of the solution, and having length at most $\alpha \sqrt{k}$ where $\alpha$ is the constant under $\mathcal{O}$ in the lemma. Since a vertex of a Voronoi diagram constructed over any subset of $B$ is equidistant from three elements from this subset and is uniquely determined by these three elements, there are at most $|B|^3$ potential locations for a Voronoi diagram vertex. Therefore there are at most $|B|^{4\alpha \sqrt{k}}$ variants for $\Gamma$. We consider only these $\Gamma$ which do not go out of the region defined by $\Delta$.

For each possible $\Gamma$, we split the instance in two parts. Denote by $Q$ the set of boxes on $\Gamma$. Let $V_1$ ($V_2$) be the set of voters inside (outside) of $\Gamma$, and $B_1$ ($B_2$) be the set of possible boxes inside (outside) of $\Gamma$, $V_1 \cup V_2 = V$, $V_1 \cap V_2 = \emptyset$, $B_1 \cup B_2 = B$. For $i \in \{1, 2\}$, let $F_i = Q \cup (F \cap B_i)$ and let $\Delta_i$ be those segments $\delta$ in $\Delta$ such that $\delta_F \in F_i$.
The fixed preference counts $\votes_1$ and $\votes_2$ for $F_1$ and $F_2$ respectively are defined as follows. For boxes in $F \setminus Q$ the value of $\votes$ is transferred directly to $\votes_1$ or $\votes_2$ depending on to which part the box went. For boxes in $F \cap Q$ and for any preference list $\sigma \in \pi(C)$ we guess how the value of $\votes(f, \sigma)$ is split between $\votes_1(f, \sigma)$ and $\votes_2(f, \sigma)$. For boxes in $Q \setminus F$ we guess separately the value of $\votes_1(f, \sigma)$ and $\votes_2(f, \sigma)$.

Finally, guess how $k + |Q \setminus F|$ is split between the two parts $k_1$ and $k_2$, $k_1$ and $k_2$ must be each at most $\frac{3}{4} k$. Also guess how $\ell - w$ is split between $\ell_1$ and $\ell_2$, where $w$ is the number of boxes in $Q \setminus F$ where the target party $p$ wins if voters from $V_1$ and $V_2$ vote according to $\votes_1$ and $\votes_2$.

Next, we run the recursion on the state $R_1=(V_1$, $B_1$, $\ell_1$, $k_1$, $F_1$, $\votes_1$, $\Delta_1 \cup \Gamma)$ and on the state $R_2=(V_2$, $B_2$, $\ell_2$, $k_2$, $F_2$, $\votes_2$, $\Delta_2 \cup \Gamma)$.
If both $R_1$ and $R_2$ are reported as valid states, we return that $R$ is valid and take $S = S_1 \cup S_2 \cup (Q \setminus F)$. Otherwise, we continue to the next choice of $\Gamma$. If for all choices of $\Gamma$ we do not succeed, we return that $R$ is not valid.


\noindent\textbf{The correctness.} Now we prove the correctness of the algorithm above. The proof is by induction on $k$. In the base case $k - |F| \le \gamma$ the algorithm tries all possible
ways to select the ballot boxes and then checks the definition directly.
In the case $k - |F| > \gamma$, assume first that the algorithm reports a given state $R=(V$, $B$, $\ell$, $k$, $F$, $\votes$, $\Delta)$ as valid. So for some $\Gamma$ the corresponding splitting states $R_1=(V_1$, $B_1$, $\ell_1$, $k_1$, $F_1$, $\votes_1$, $\Delta_1 \cup \Gamma)$
and $R_2=(V_2$, $B_2$, $\ell_2$, $k_2$, $F_2$, $\votes_2$, $\Delta_2 \cup \Gamma)$ are reported as valid by the algorithm, and by induction that means that they are actually valid states since both $k_1$ and $k_2$ are at most $\frac{3}{4}k$. Consider sets $S_1 \subset B_1 \setminus F_1$ and $S_2 \subset B_2 \setminus F_2$ returned by the recursive calls, by induction $S_1$ and $S_2$ also satisfy Definition~\ref{def:valid} on the validity of $R_1$ and $R_2$, respectively. As before, denote $S_1 \cup S_2 \cup (Q \setminus F)$ by $S$. We claim that $S$ satisfies Definition~\ref{def:valid} for $R$, and therefore $R$ is valid.

By construction of $R_1$ and $R_2$, $k = k_1 + k_2 - |Q \setminus F|$, and since $S_1$, $S_2$, and $Q \setminus F$ are pairwise disjoint, the size of $S$ is indeed equal to $k - |F|$. Next, we show that in the election with boxes $S \cup F$ the target party $p$ wins in at least $\ell$ districts of $S$ and that the votes in districts of $F$ are distributed according to $v$.
The next claim shows that for voters in each part the box where they vote is the same in the election with boxes $S_i \cup F_i$ and in the election with boxes $S \cup F$.

\begin{claim}
    For $i \in \{1, 2\}$ and any voter $x \in V_i$, the closest box to $x$ in $S \cup F$ belongs to $S_i \cup F_i$ as well.
\label{claim:closest_preserved}
\end{claim}
\begin{proof}
    Let $i = 1$, the other case is analogous. Assume the contrary, there is a voter $x \in V_1$ such that there is a box $b \in S_2 \cup F_2 \setminus Q$ which is strictly closer to $x$ than any box in $S_1 \cup F_1$, since two boxes cannot be at the same distance to $x$. Consider a straight line segment from $x$ to $b$, since $x$ and $b$ lie on the different sides of $\Gamma$,
    the segment crosses $\Gamma$ at least once, denote any intersection point by $y$.
    Since $R_2$ is valid, the segments of $\Gamma$ lie completely inside the cells of $Q$ in the Voronoi diagram of $S_2 \cup F_2$.
    That means that there is $q \in Q$ such that $y$ is at least as close to $q$ as to any other box in $S_2 \cup F_2$, including $b$. But then $\dist(x, q) \le \dist(x, y) + \dist(y, q) \le \dist(x, y) + \dist(y, b)=\dist(x,b)$, where the first inequality is the triangle inequality, and this contradicts the initial assumption.
\end{proof}

By Claim~\ref{claim:closest_preserved}, for each $i \in \{1, 2\}$ and for each box in $S_i$ exactly the same set of voters votes there in the election with boxes $S \cup F$ compared to the election with boxes $S_i \cup F_i$. It means that among the boxes in $S_1$ and $S_2$ the target party $p$ wins in exactly $\ell_1 + \ell_2$ of them. For each $i \in \{1, 2\}$ and each box $f$ in $F_i \setminus Q$ it also holds that exactly the same set of voters votes in $f$ in the election with boxes $S \cup F$ as votes in $f$ in the election with boxes $S_i \cup F_i$.
Since $v(f, \sigma) = v_i(f, \sigma)$ for any $\sigma \in \pi(C)$, the vote distribution on $F_i \setminus Q$ is as desired. For each box $q \in Q$ and each $\sigma \in \pi(C)$ by construction
$v(q, \sigma) = v_1(q, \sigma) + v_2(q, \sigma)$. Since voters from $V_1$ and $V_2$ who vote in $q$ are exactly preserved, $v(q, \sigma)$ is indeed equal to the number of voters with preference list $\sigma$ who vote in $q$. This finishes the proof that $R$ is valid and $S$ satisfies Definition~\ref{def:valid}.

Finally, we show that for each $\delta \in \Delta$, $\delta$ lies inside the cell of $\delta_F$ in the Voronoi diagram of $S \cup F$. If $\delta_F \in Q$, since $R_1$ and $R_2$ are valid, $\delta$ lies inside the cell of $\delta_F$ in the Voronoi diagram of $S_i \cup F_i$ for $i \in \{1, 2\}$, then no other point in $S \cup F$ is closer to each point on $\delta$ than $\delta_F$. If $\delta_F \notin Q$, $\delta$ is completely inside or outside of $\Gamma$, and the same argument as in Claim~\ref{claim:closest_preserved} shows that for any point on $\delta$ the closest box is preserved, then it has to be $\delta_F$, since $\delta$ is in $\Delta_i$ for some $i \in \{1, 2\}$, and $R_i$ is valid.

Towards the other direction, assume that $R$ is valid, we show that the algorithm correctly reports the validity of $R$. Consider a particular $S$ from Definition~\ref{def:valid}, by Lemma~\ref{lemma:noose} for the Voronoi diagram of $S \cup F$, there exists a polygon $\Gamma$ of length $O(\sqrt{k})$ which is a balanced separator for $S \cup F$.
Since the algorithm tries all polygons of this form, it will try $\Gamma$ as well, or report that $R$ is valid earlier. When considering the polygon $\Gamma$, the algorithm will guess the right values of $k_1$, $k_2$, $\ell_1$, $\ell_2$, and $v_1$, $v_2$ on the boxes of $Q \setminus F$, according to the elections on $S \cup F$.
The new states $R_1$ and $R_2$ are valid since the elections on $S \cap F$ exactly induce conditions on $R_i$, and there exists a valid selection of boxes $S_i = (S \cup F) \cap B_i \setminus Q$, for each $i \in \{1, 2\}$.  Therefore the recursive call will find $R_i$ valid by induction. Thus the proof of correctness is finished.

\noindent\textbf{Running time analysis.} Denote by $T(k)$ the worst-case running time of our algorithm where $k$ is the respective parameter of the input state. If $k \le \gamma$, the algorithm does a brute force in polynomial time so $T(k) = (n + m)^{\OhOp{1}}$. If $k > \gamma$, we try at most $m^{4\alpha\sqrt{k}}$ polygons $\Gamma$, for each of them we in time at most $k \ell n^{2\alpha \sqrt{k}}$ guess the parameters of the new two instances $k_1$, $\ell_1$, $\votes_1$ and $k_2$, $\ell_2$, $\votes_2$.
We run our algorithm on the two instances recursively, and in both instances the value of the parameter is bounded by $\frac{3}{4}k$, so we get the following recurrence relation,
\[\begin{cases}
    T(k) \le (n + m)^{\beta \sqrt{k}} T(\frac{3}{4} k), \quad k > \gamma\\
    T(k) \le (n + m)^\beta, \quad k \le \gamma,
\end{cases}\]
where $\beta$ is some constant. Now, by applying the first upper bound until the expression under $T$ is at most $\gamma$ we get
\begin{multline*}
    T(k) \le (n + m)^{\beta \sqrt{k} (1 + q + q^2 + \ldots + q^t)} \\\le (n + m)^{\beta \sqrt{k} /(1 - q)} = (n + m)^{\OhOp{\sqrt{k}}},
\end{multline*}
where $q = \sqrt{3/4}$ and $t = \lceil \log_{4/3} (k / \gamma) \rceil$.

\end{proof}

\section{Concluding Remarks}

We initiated study of a gerrymandering problem from the perspective of fine-grained and parameterized complexity. Our main result is asymptotically tight algorithm with respect to parameter $k$, the number of opened ballot boxes, together with a matching lower bound. 

For future work, it would be interesting to investigate the complexity of finding an approximate solution for the problem. That is the question whether there is an efficient (polynomial time) algorithm that either finds a set of $k$ ballot boxes such that $p$ wins $\frac{\ell}{c}$ districts, for some constant $c$, or correctly decides that in no elections with $k$ districts the candidate $p$ wins $\ell$ of them. Another interesting question is what happens, when we introduce obstacles, such as lakes or hills, in the instance so that the distances are not anymore Euclidean distances in the plane. Can we still get better than trivial $m^{k}$ algorithm in this case? Finally, it is also natural to impose restrictions on sizes of districts, so they are not too disproportional. 

\section{Acknowledgments}

This work is supported by the Research Council of Norway
via the project ``MULTIVAL''.

\bibliography{references}
\bibliographystyle{plainurl}
\end{document}